\documentclass{llncs}
\usepackage[dvips]{graphicx}
\usepackage{pstricks}
\usepackage{psfrag}
\usepackage{amssymb}
\begin{document}

\mainmatter

\title{Obtaining a Planar Graph by Vertex Deletion}
\author{D\'aniel Marx \and Ildik\'o Schlotter}

\institute{
Department of Computer Science and Information Theory, \\
Budapest University of Technology and Economics, \\
Budapest, H-1521, Hungary. \\
\email{\{dmarx,ildi\}@cs.bme.hu}\\
}

\maketitle

\begin{abstract}
In the $k$-\textsc{Apex} problem the task is to find at most $k$ vertices
whose deletion makes the given graph planar.
The graphs for which there exists a solution form a minor closed class of graphs,
hence by the deep results of Robertson and Seymour \cite{sey95,sey04}, there is an $O(n^3)$ time algorithm
for every fixed value of $k$. However, the proof is extremely complicated and the constants hidden by
the big-O notation are huge.
Here we give a much simpler algorithm for this problem with quadratic running time, by
iteratively reducing the input graph and then applying techniques for graphs of bounded treewidth.

\noindent {\bf Keywords:} Planar graph, Apex graph, FPT algorithm, Vertex deletion.
\end{abstract}

\section{Introduction}

Planar graphs are subject of wide research interest in graph theory.
There are many generally hard problems which can be solved in polynomial time when considering planar graphs,
e.g., \textsc{Maximum Clique}, \textsc{Maximum Cut}, and \textsc{Subgraph Isomorphism}
\cite{epp99,had75}.
For problems that remain NP-hard on planar graphs, we often have efficient approximation algorithms.
For example, the problems \textsc{Independent Set, Vertex Cover}, and  \textsc{Dominating Set}
admit an efficient linear time approximation scheme \cite{bak94,lip80}.
The research for efficient algorithms for problems on planar graphs is still very intensive.

Many results on planar graphs can be extended to almost planar graphs,
which can be defined in various ways.
For example, we can consider possible embeddings of a graph in a surface other than the plane.
The genus of a graph is the minimum number of handles that must be added to the
plane to embed the graph without any crossings.
Although determining the genus of a graph is NP-hard \cite{tho89},
the graphs with bounded genus are subjects of wide research.
A similar property of graphs is their crossing number, i.e.,
the minimum possible number of crossings with which the graph can be drawn in the plane.
Determining crossing number is also NP-hard \cite{gar83}.

In \cite{cai96} Cai introduced another notation to capture the distance of a
graph $G$ from a graph class $\mathcal{F}$, based on the number of certain elementary modification steps.
He defines the distance of $G$ from $\mathcal{F}$ as the minimum number of modifying steps needed to
make $G$ a member of $\mathcal{F}$.
Here, modification can mean the deletion or addition of edges or vertices.
In this paper we consider the following question: given a graph $G$ and an integer
$k$, is there a set of at most $k$ vertices in $G$, whose deletion makes $G$ planar?

It was proven by Lewis and Yannakakis in \cite{lew80} that the node-deletion problem is NP-complete
for every non-trivial hereditary graph property decidable in polynomial time.
As planarity is such a property, the problem of finding a maximum induced planar subgraph is NP-complete,
so we cannot hope to find a polynomial-time algorithm that answers the above question.
Therefore, following Cai, we study the problem in the framework of parameterized complexity
developed by Downey and Fellows \cite{dow99}.
This approach deals with problems in which besides the input $I$ an integer $k$ is also
given. The integer $k$ is referred to as the parameter.
In many cases we can solve the problem in time $O(n^{f(k)})$.
Clearly, this is also true for the problem we consider.
Although this is polynomial time for each fixed $k$, these algorithms are practically
too slow for large inputs, even if $k$ is relatively small.
Therefore, the standard goal of parameterized analysis is to take the parameter
out of the exponent in the running time.
A problem is called \emph{fixed-parameter tractable} (FPT) if it can be solved in
time $O(f(k)p(|I|))$, where $p$ is a polynomial not depending on $k$, and $f$ is an arbitrary function.
An algorithm with such a running time is also called FPT.
For more on fixed-parameter tractability see e.g.  \cite{dow99}, \cite{nie06} or \cite{flu06}.

The standard parameterized version of our problem is the
following: given a graph $G$ and a parameter $k$, the task is to
decide whether deleting at most $k$ vertices from $G$ can result in a planar graph.
Such a set of vertices is sometimes called a set of \emph{apex vertices} or \emph{apices},
so we will denote the class of graphs for which the answer is `yes' by $\textup{Apex}(k)$.
We note that Cai \cite{cai96} used the notation $\textup{Planar} + kv$ to denote this class.

In the parameterized complexity literature, numerous similar node-deletion problems have been studied.
A classical result of this type by Bodlaender \cite{bod94} and Downey and Fellows \cite{dow92} states that
the \textsc{Feedback Vertex Set} problem, asking whether a graph can be made acyclic by the deletion of at most $k$ vertices,
is FPT. The parameterized complexity of the directed version of this problem has been a long-standing open question,
and it has only been proved recently that it is FPT as well \cite{che07}.
Fixed-parameter tractability has also been proved for the problem of finding $k$ vertices whose deletion results
in a bipartite graph \cite{ree04}, or in a chordal graph \cite{marxx}.
On the negative side, the corresponding node-deletion problem for wheel-free graphs was proved to be W[2]-hard \cite{lok08}.

Considering the graph class $\textup{Apex}(k)$,
we can observe that this family of graphs is closed under taking minors.
The celebrated graph minor theorem by Robertson and Seymour states that such families can be characterized by
a set of excluded minors \cite{sey04}. They also showed that for each graph $H$ it can be tested in cubic time whether
a graph contains $H$ as a minor \cite{sey95}. As a consequence, membership for such graph classes can be decided in cubic time.
In particular, we know that there exists an algorithm with running time $O(f(k)n^3)$ that
can decide whether a graph belongs to $\textup{Apex}(k)$.
However, the proof of the graph minor theorem is non-constructive in the following sense.
It proves the existence of an algorithm for the membership test that uses the excluded minor
characterization of the given graph class, but does not provide any algorithm for determining this characterization.
Recently, an algorithm was presented in \cite{adl08} for constructing the set of excluded minors
for a given graph class closed under taking minors,
which yields a way to explicitly construct the algorithm whose existence was proved by Robertson and Seymour.
We remark that it follows also from \cite{fel89} that an algorithm for testing membership in
$\textup{Apex}(k)$ can be constructed explicitly.

Although these results provide a general tool that can be applied to our specific problem,
no direct FPT algorithm has been proposed for it so far.
In this paper we present an algorithm which decides membership for $\textup{Apex}(k)$ in $O(f(k)n^2)$ time.
Note that the presented algorithm runs in quadratic time, and hence yields a better running time than
any algorithm using the minor testing algorithm that is applied in the above mentioned approaches.
Moreover, if $G \in \textup{Apex}(k)$ then our algorithm also returns a solution,
i.e., a set $S \in V(G)$, $|S| \leq k$ such that $G-S$ is planar.

The presented algorithm is strongly based on the ideas used by Grohe in \cite{gro04}
for computing crossing number.
Grohe uses the fact that the crossing number of a graph is an upper bound for its genus.
Since the genus of a graph in $\textup{Apex}(k)$ cannot be bounded by a function of $k$,
we need some other ideas.
As in \cite{gro04}, we exploit the fact that in a graph with large treewidth
we can always find a large grid minor \cite{sey94}. Examining the structure of
the graph with such a grid minor, we can reduce our problem to a smaller instance.
Applying this reduction several times, we finally get an instance with
bounded treewidth. Then we make use of Courcelle's Theorem \cite{cou90}, which
states that every graph property that is expressible in monadic second-order logic
can be decided in linear time on graphs of bounded treewidth.

The paper is organized as follows.
Section \ref{notation} summarizes our notation,
Sect. \ref{algo_sketch} outlines the algorithm,
Sect. \ref{phase_I} and \ref{phase_II} describe the two phases of the algorithm.

\section{Notation}
\label{notation}

Graphs in this paper are assumed to be simple, since both loops and multiple edges are
irrelevant in the
$k$-\textsc{Apex}
problem.
The vertex set and edge set of a graph $G$
are denoted by $V(G)$ and $E(G)$, respectively.
The edges of a graph are unordered pairs of its vertices.
If $G'$ is a subgraph of $G$ then $G-G'$ denotes
the graph obtained by deleting $G'$ from $G$. For a set of vertices $S$ in $G$,
we will also use $G-S$ to denote the graph obtained by deleting $S$ from $G$.

A graph $H$ is a \emph{minor} of a graph $G$ if it can be obtained from a subgraph of $G$
by contracting some of its edges. Here \emph{contracting an edge} $e$ with endpoints $a$ and $b$ means
deleting $e$, and then identifying vertices $a$ and $b$.

A graph $H$ is a \emph{subdivision} of a graph $G$ if $G$ can be obtained from $H$
by contracting some of its edges that have at least one endpoint of degree two.
Or, equivalently, $H$ is a subdivision of $G$ if $H$ can be obtained from $G$ by
replacing some of its edges with newly introduced paths
such that the inner vertices of these paths have degree two in $H$.
We refer to these paths in $H$ corresponding to edges of $G$ as \emph{edge-paths}.
A graph $H$ is a \emph{topological minor} of $G$ if $G$ has a subgraph that is a subdivision of $H$.
We say that $G$ and $G'$ are \emph{topologically isomorphic} if they both are subdivisions of a graph $H$.

The $g \times g$ \emph{grid} is the graph $G_{g \times g}$ where
$V(G_{g \times g}) = \{v_{ij} \,|\, 1 \leq i,j \leq g \}$ and
$E(G_{g \times g}) = \{ v_{ij} v_{i'j'} \,|\, |i-i'|+|j-j'| = 1 \}$.
Instead of giving a formal definition for the \emph{hexagonal grid} of radius $r$,
which we will denote by $H_r$, we refer to the illustration shown in Fig. \ref{fig_hexa}.
A \emph{cell} of a hexagonal grid is one of its cycles of length $6$.

\begin{figure}[t]
\begin{center}
\includegraphics[scale=0.25]{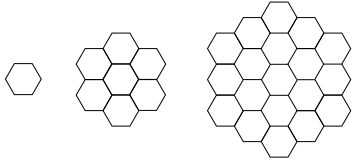}
\end{center}
\caption{The hexagonal grids $H_1$, $H_2$, and $H_3$.}
\label{fig_hexa}
\end{figure}

A \emph{tree decomposition} of a graph $G$ is a pair $(T,(V_t)_{t \in V(T)})$ where
$T$ is a tree, $V_t \subseteq V(G)$ for all $t \in V(T)$, and the following are true:

\begin{itemize}
\item
for all $v \in V(G)$ there exists a $t \in V(T)$ such that $v \in V_t$,
\item
for all $xy \in E(G)$ there exists a $t \in V(T)$ such that $x,y \in V_t$,
\item
if $t$ lies on the path connecting $t'$ and $t''$ in $T$, then
$V_t \supseteq V_{t'} \cap V_{t''}$.
\end{itemize}

The \emph{width} of such a tree decomposition is the maximum of $|V_t|-1$ taken over all $t \in V(T)$.
The \emph{treewidth} of a graph $G$, denoted by $\textup{tw}(G)$,
is the smallest possible width of a tree decomposition of $G$.
For an introduction to treewidth see e.g. \cite{bod97,die00}.

\section{Problem Definition and Overview of the Algorithm}
\label{algo_sketch}
We are looking for the solution of the following problem:

\begin{center}
\fbox{ \parbox{11cm}{
$k$-\textsc{Apex}
problem: \\[2pt]
\begin{tabular}{p{1cm}p{9.5cm}}
Input: & A graph $G=(V,E)$ and an integer $k$. \\
Task: & Find a set $X$ of at most $k$ vertices in $V$ such that $G - X$ is planar.
\end{tabular}
}}
\end{center}

Here we give an algorithm $\mathcal{A}$ which solves this problem
in time $O(f(k)n^2)$ for some function $f$, where $n$ is the number of
vertices in the input graph.
Algorithm $\mathcal{A}$ works in two phases.
In the first phase (Sect. \ref{phase_I}) we compress the given graph repeatedly, and finally
either conclude that there is no solution for our problem or construct an equivalent problem instance
with a graph having bounded treewidth. In the latter case we solve the problem in the second phase
of the algorithm (Sect. \ref{phase_II}) by applying
Courcelle's Theorem which gives a linear time algorithm for the evaluation of MSO-formulas on
bounded treewidth graphs.

To describe the first step of our algorithm, we need some deep results from graph minor theory.
The following result states that every graph having large treewidth
must contain a large grid as a minor.

\begin{theorem}[Excluded Grid Theorem, \cite{sey86}]
\label{thm_egt}
For every fixed integer  $r$ there exists an integer $w(r)$ such that
if $\textup{tw}(G) > w(r)$  then $G$ contains $G_{r \times r}$ as a minor.
\end{theorem}

The grid minor guaranteed by this theorem in the case when the treewidth of the graph $G$ is large
can be found in cubic time. However, we need a linear-time algorithm for finding a large grid minor, 
so we have to make use of the following result, which states that
if the graph is planar, then the bound on $w(r)$ is linear:

\begin{theorem}[Excluded Grid Theorem for Planar Graphs, \cite{sey94}]
\label{thm_egt_planar}
For every integer $r$ and every planar graph $G$,
if $\textup{tw}(G) > 6r-5$ then $G$ contains $G_{r \times r}$ as a minor.
\end{theorem}

Also, we will use the following algorithmic results:

\begin{theorem}[\cite{bod96,per00}]
\label{thm_treedecomp}
For every fixed integer $w$ there exists a linear-time algorithm that, given a graph $G$, does the following:
\begin{itemize}
\item either produces a tree decomposition of $G$ of width at most $w$, or
\item outputs a subgraph $G'$ of $G$ with $\textup{tw}(G')>w$,
together with a tree decomposition of $G'$ of width at most $2w$.
\end{itemize}
\end{theorem}

\begin{theorem}[\cite{arn91}]
\label{thm_minortest}
For every fixed graph $H$ and integer $w$ there exists a linear-time algorithm that, given a graph $G$
and a tree decomposition for $G$ of width $w$, returns a minor of $G$ isomorphic to $H$, if this is possible.
\end{theorem}

Now, we are ready to state our first lemma, which provides the key structures for the mechanism of our algorithm.
In this lemma, we focus on hexagonal grids instead of rectangular grids.
The reason for this is the well-known fact that if a graph of maximum degree three is a minor of another graph,
then it is also contained in it as a topological minor.
This property of the hexagonal grid will be very useful later on.

\begin{lemma}
\label{alg_b}
For every pair of fixed integers $r$ and $k$ there is a linear-time algorithm $\mathcal{B}$,
that, given an input graph $G$, does the following:
\begin{itemize}
\item
either produces a tree decomposition of $G$ of width $w(r,k)=24r-11+k$, or
\item
finds a subdivision of $H_r$ in $G$, or
\item
correctly concludes that $G \notin \textup{Apex}(k)$.
\end{itemize}
\end{lemma}

\begin{proof}
Let $r$ and $k$ be arbitrary fixed integers.
Let us run the algorithm provided by Theorem~\ref{thm_treedecomp} for $w=w(r,k)$ on graph $G$.
If it produces a tree decomposition of width $w(r,k)$ for $G$, then we output it.
Otherwise let $G'$ be the subgraph of $G$ with $\textup{tw}(G')>w(r,k)$ that has been provided
together with a tree decomposition $\mathcal{T}'$ for it having width at most $2w(r,k)$.

Clearly, if $G' \notin \textup{Apex}(k)$, then $G \notin \textup{Apex}(k)$ also holds as $G'$ is a subgraph of $G$.
On the other hand, if $G' \in \textup{Apex}(k)$, then there exists a set $S \subseteq V(G)$ with $|S| \leq k$
such that $G'-S$ is planar. Deleting a vertex of a graph can only decrease its treewidth by at most one,
so $\textup{tw}(G'-S) > w(r,k)-k=6(4r-1)-5$. Now, Theorem~\ref{thm_egt_planar} implies that
$G'-S$ contains $G_{(4r-1) \times (4r-1)}$ as a minor.
Since the hexagonal grid with radius $r$ is a subgraph of the $(4r-1) \times (4r-1)$ grid,
we get that $G'-S$ must also contain $H_r$ as a minor, and hence as a topological minor.

Thus, we get that either $G \notin \textup{Apex}(k)$, or $G'$ (and hence $G$)
contains $H_r$ as a (topological) minor. Now, using the algorithm of
Theorem~\ref{thm_minortest} for $G'$ and $\mathcal{T}'$,
we can find a subgraph of $G'$ isomorphic to a subdivision of $H_r$ in linear time, if possible.
If the algorithm produces such a subgraph, then we output it, otherwise we can correctly conclude
that $G \notin \textup{Apex}(k)$.
\qed
\end{proof}

In algorithm $\mathcal{A}$ we will run $\mathcal{B}$ several times.
As long as the result is a hexagonal grid of radius $r$ as topological minor,
we will run Phase I of algorithm $\mathcal{A}$, which compresses the graph $G$.
If at some step algorithm $\mathcal{B}$ gives us a tree
decomposition of width $w(r,k)$, we run Phase II.
(The constant $r$ will be fixed later.) And of course if at some step
$\mathcal{B}$ finds out that $G \notin \textup{Apex}(k)$, then
algorithm $\mathcal{A}$ can stop with the output ``No solution.''

Clearly, we can assume without loss of generality that the input graph is simple,
and it has at least $k+3$ vertices.
So if $G \in \textup{Apex}(k)$, then deleting $k$ vertices from $G$
(which means the deletion of at most $k(|V(G)|-1)$ edges) results in a planar graph,
which has at most $3|V(G)|-6$ edges.
Therefore, if $|E(G)|>(k+3)|V(G)|$ then surely $G \notin \textup{Apex}(k)$.
Since this can be detected in linear time, we can assume that $|E(G)|\leq (k+3)|V(G)|$.

\section{Phase I of Algorithm $\mathcal{A}$}
\label{phase_I}

In Phase I we assume that after running $\mathcal{B}$ on $G$ we
get a subgraph $H'_r$ that is a subdivision of $H_r$.
Our goal is to find a set of vertices $X$ such that $G-X$ is planar, and $|X| \leq k$.
Let $\textup{ApexSets}(G,k)$ denote the family of sets of vertices that have these
properties, i.e., let
$\textup{ApexSets}(G,k)=\{ X \subseteq V(G) \, | \, |X| \leq k$ and
$G-X$ is planar$\}$.
Since the case $k=1$ is very simple we can assume that $k > 1$.

\textbf{Reduction A: Flat zones.}
In the following we regard the grid $H'_r$ as a fixed subgraph of $G$.
Let us define $z$ zones in it. Here $z$ is a constant
depending only on $k$, which we will determine later.
A \emph{zone} is a subgraph of $H'_r$ which is topologically isomorphic
to the hexagonal grid $H_{2k+5}$. We place such zones next to each other in the well-known
radial manner with radius $q$, i.e., we replace each hexagon of $H_q$ with a subdivision of $H_{2k+5}$.
It is easy to show that in a hexagonal grid with radius ${(q-1)(4k+9)+(2k+5)}$
we can define this way $3q(q-1)+1$ zones that only intersect in their outer circles.
So let $r=(q-1)(4k+9)+(2k+5)$, where we choose $q$ big enough to get at least $z$ zones,
i.e., $q$ is the smallest integer such that $3q(q-1)+1 \geq z$.
Let the set of these innerly disjoint zones be $\mathcal{Z}$, and the
subgraph of these zones in $H'_r$ be $R$.

Let us define two types of \emph{grid-components}.
An edge which is not contained in $R$ is a grid-component if it connects
two vertices of $R$. A subgraph of $G$ is a grid-component if it is a (maximal) connected component of $G-R$.
A grid-component $K$ is \emph{attached} to a vertex $v$ of the grid $R$ if it has a vertex adjacent to $v$, or
(if $K$ is an edge) one of its endpoints is $v$.
The \emph{core} of a zone is the (unique) subgraph of the zone which is topologically
isomorphic to $H_{2k+3}$ and lies in the middle of the zone.
Let us call a zone $Z \in \mathcal{Z}$ \emph{open} if there is a vertex in its core that is connected to
a vertex $v$ of another zone in $\mathcal{Z}$, $v \notin V(Z)$, through a grid-component.
A zone is \emph{closed} if it is not open.

\begin{figure}[t]
\begin{center}
    \psfrag{a}{(a)}
    \psfrag{b}{(b)}
\includegraphics[scale=0.5]{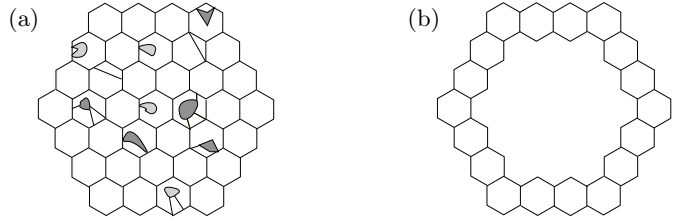}
\end{center}
\caption{(a) An induced subgraph of a flat zone, together with its grid components.
Among them, there are two edges, four edge-components (shown in light gray) and five cell-components (dark gray).
(b) The ring $R_3$ of $Z$.}
\label{fig_comp_ring}
\end{figure}

For a subgraph $H$ of $R$ we let  $T(H)$ denote the subgraph of $G$ induced by the vertices of $H$ and
the vertices of the grid-components which are only attached to $H$.
Let us call a zone $Z$ \emph{flat} if it is closed and $T(Z)$ is planar.
Let $Z$ be such a flat zone.
See Fig.~\ref{fig_comp_ring} (a) for an illustration of a flat zone together with its grid-components.
A grid-component is an \emph{edge-component} if it is
either only attached to one edge-path  of $Z$ or only to one vertex of $Z$.
Otherwise, it is a \emph{cell-component} if it is only attached to vertices of one cell.
As a consequence of the fact that all embeddings of a 3-connected graph are equivalent
(see e.g. \cite{die00}), and $Z$ is a subdivision of such a graph,
every grid-component attached to some vertex in the core of $Z$ must be one of these two types.
Note that we can assume that in an embedding of $T(Z)$ in the plane,
all edge-components are embedded in an arbitrarily small neighborhood of
the edge-path (or vertex) which they belong to.

Let us define the \emph{ring} $R_i$ ($1 \leq i \leq 2k+4$) as the union of those cells in $Z$
that have common vertices both with the $i$-th and the $(i+1)$-th concentrical circle of $Z$.
Let $R_0$ be the cell of $Z$ that lies in its center.
The zone $Z$ can be viewed as the union of $2k+5$ concentrical rings,
i.e., the union of the subgraphs $R_i$ for $0 \leq i \leq 2k+4$.
Figure~\ref{fig_comp_ring} (b) depicts the ring $R_3$.

\begin{lemma}
\label{equivalence}
Let $Z$ be a flat zone in $R$, and let $G'$ denote the graph $G-T(R_0)$.
Then $X \in \textup{ApexSets}(G',k)$ implies $X \in \textup{ApexSets}(G,k)$.
\end{lemma}

\begin{proof}
Suppose $X \in \textup{ApexSets}(G',k)$.
Since $G-T(R_0)-X$ is planar, we can fix a planar embedding $\phi$ of it.
If $R_i \cap X = \emptyset$ for some $i$ ($2 \leq i \leq 2k+2$)
then let $W_i$ denote the maximal subgraph of $G-T(R_0)-X$ for
which $\phi(W_i)$ is in the region determined by $\phi(R_i)$ (including $R_i$).
If $R_i \cap X$ is not empty then let $W_i$ be the empty graph.
Note that if $2 \leq i \leq 2k$ then $W_i$ and $W_{i+2}$ are disjoint.
Therefore, there exists an index $i$ for which $W_i \cap X = \emptyset$
and $W_i$ is not empty.
Let us fix this $i$.

Let $Q_i$ denote $T(\bigcup_{j=0}^i R_j)$.
We prove the lemma by giving an embedding for $G-X'$ where
$X'=X \setminus V(Q_{i-1})$.
The region $\phi(R_i)$ divides the plane in two other regions.
As $Z$ is flat, vertices of $Q_{i-1}$ can only be adjacent to vertices of $Q_i$.
Thus we can assume that in the finite region only vertices of $Q_{i-1}$ are embedded, so
$G-X'-(Q_{i-1} \cup W_i)$ is entirely embedded in the infinite region.
Let $U$ denote those vertices in $Q_{i-1}$ which are adjacent
to some vertex in $G-Q_{i-1}$.
Observe that the vertices of $U$ lie on the $i$-th concentrical circle of $Z$, hence,
the restriction of $\phi$ to $G-X'-(Q_{i-1}-U)$ has a face whose boundary contains $U$.

Now let $\theta$ be a planar embedding of $T(Z)$, and let us restrict $\theta$ to $Q_{i-1}$.
Note that $U$ only contains vertices which are either
adjacent to some vertex in $R_i$ or are adjacent to cell-components
belonging to a cell of $R_i$.
But $\theta$ embeds $R_i$ and its cell-components also, and therefore the restriction of
$\theta$ to $Q_{i-1}$ results in a face whose boundary contains $U$. Here we used also
that $R_i$ is a subdivision of a 3-connected graph whose embeddings are equivalent.

Now it is easy to see that we can combine $\theta$ and $\phi$ in such a way
that we embed $G-X'-(Q_{i-1}-U)$ according to $\phi$
and, similarly, $Q_{i-1}$ according to $\theta$, and then ``connect'' them by identifying
$\phi(u)$ and $\theta(u)$ for all $u \in U$. This gives the desired embedding of
$G-X'$. Finally, we have to observe that
$X' \in \textup{ApexSets}(G,k)$ implies
$X \in \textup{ApexSets}(G,k)$, since $X' \subseteq X$ and $|X| \leq k$.
\qed
\end{proof}

This lemma has a trivial but crucial consequence:
$X \in \textup{ApexSets}(G,k)$ if and only if
$X \in \textup{ApexSets}(G-T(R_0),k)$,
so deleting $T(R_0)$ reduces our problem to an equivalent instance.
Let us denote this deletion as \emph{Reduction A}.

Note that the closedness of a zone $Z$ can be decided
by a simple breadth first search, which can also produce the graph $T(Z)$.
Planarity can also be tested in linear time \cite{hop74}.
Therefore we can test whether a zone is flat, and if so, we can apply Reduction A on it in linear time.

Later we will see that unless there are some easily recognizable vertices in our graph
which must be included in every solution, a flat zone can always be found (Lemma \ref{zones}).
This yields an easy way to handle graphs with large treewidth:
compressing our graph by repeatedly applying Reduction A we can reduce the problem
to an instance with bounded treewidth.

\textbf{Reduction B: Well-attached vertices.}
A subgraph of $R$ is a \emph{block} if it is topologically isomorphic to $H_{k+3}$.
A vertex of a given block is called \emph{inner vertex}
if it is not on the outer circle of the block.

\begin{figure}[t]
\begin{center}
    \psfrag{Bx}{$B_x$}
    \psfrag{By}{$B_y$}
    \psfrag{Cx}{$C_x$}
    \psfrag{Cy}{$C_y$}
    \psfrag{x}{$x$}
    \psfrag{y}{$y$}
    \psfrag{ax}{$a_x$}
    \psfrag{ay}{$a_y$}
    \psfrag{P}{$P$}
    \psfrag{P'}{$P'$}
\includegraphics[scale=0.5]{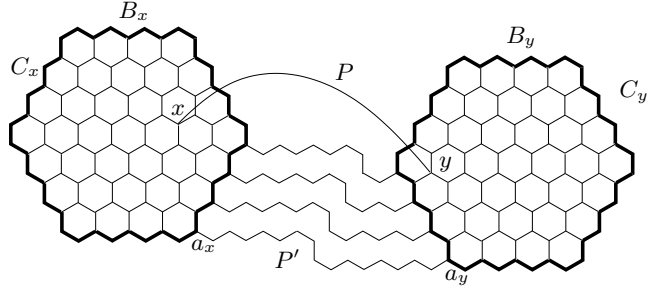}
\end{center}
\caption{Illustration for Lemma~\ref{blocks}. The edges of $C_x$ and $C_y$ are shown in bold.}
\label{fig_blocks}
\end{figure}

\begin{lemma}
\label{blocks}
Let $X \in \textup{ApexSets}(G,k)$.
Let $x$ and $y$ be inner vertices of the disjoint blocks $B_x$ and $B_y$, respectively.
If $P$ is an $x-y$ path that (except for its endpoints) doesn't contain any vertex from $B_x$ or $B_y$, then
$X$ must contain a vertex from $B_x$, $B_y$ or $P$.
\end{lemma}

\begin{proof}
See Fig.~\ref{fig_blocks} for the illustration of this proof.
Let $C_x$ and $C_y$ denote the outer circle of $B_x$ and $B_y$, respectively.
Let us notice that since $B_x$ and $B_y$ are disjoint blocks, there exist at least
$k+3$ vertex disjoint paths between their outer circles, which---apart from their endpoints---do
 not contain vertices from $B_x$ and $B_y$.
Moreover, it is easy to see that these paths can be defined in a way such that
their endpoints that lie on $C_x$ are on the border of different cells of $B_x$.
To see this, note that the number of cells which lie on the border of a given block is $6k+12$.
At least three of these paths must be in $G-X$ also.
Since $x$ can lie only on the border of at most two cells having common vertices with $C_x$, we
get that there is a path $P'$ in $G-X$ whose endpoints are $a_x$ and $a_y$
(lying on $C_x$ and $C_y$, resp.), and there exist no cell of $B_x$ whose border
contains both $a_x$ and $x$.

Let us suppose that $B_x \cup B_y \cup P$ is a subgraph of $G-X$.
Since all embeddings of a 3-connected planar graph are equivalent,
we know that if we restrict an
arbitrary planar embedding of $G-X$ to $B_x$, then all faces having $x$ on their border
correspond to a cell in $B_x$. Since $x$ and $y$ are connected through $P$
and $V(P) \cap V(B_x)= \{x\}$, we get that $y$ must be embedded in a region $F$ corresponding to a cell
$C_F$ of $B_x$. But this implies that $B_y$ must entirely be embedded also in $F$.

Since $V(P'-a_x-a_y) \cap V(B_x) =\emptyset$ and $P'$ connects $a_x \in V(B_x)$
and $a_y \in V(B_y)$ we have that $a_y$ must lie on the border of $F$.
But then $C_F$ is a cell of $B_x$ containing both $a_x$ and $x$ on its border,
which yields the contradiction.
\qed
\end{proof}

Using this lemma we can identify certain vertices that have to be deleted.
Let $x$ be a \emph{well-attached} vertex in $G$ if there exist
paths $P_1, P_2, \dots, P_{k+2}$ and disjoint blocks $B_1, B_2, \dots, B_{k+2}$ such that
$P_i$ connects $x$ with an inner vertex of $B_i$ ($1 \leq i \leq k+2$),
the inner vertices of $P_i$ are not in $R$,
and if $i \neq j$ then the only common vertex of $P_i$ and $P_j$ is $x$.

\begin{figure}[t]
\begin{center}
    \psfrag{x}{$x$}
    \psfrag{B1}{$B_1$}
    \psfrag{B2}{$B_2$}
    \psfrag{Bk+2}{$B_{k+2}$}
    \psfrag{P1}{$P_1$}
    \psfrag{P2}{$P_2$}
    \psfrag{Pk+2}{$P_{k+2}$}
\includegraphics[scale=0.35]{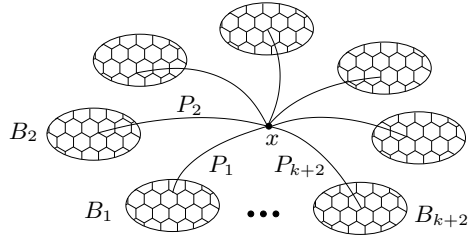}
\end{center}
\caption{A well-attached vertex.}
\label{fig_wellattach}
\end{figure}

\begin{lemma}
\label{well-attached}
Let $X \in \textup{ApexSets}(G,k)$.
If $x$ is well-attached then $x \in X$.
\end{lemma}

\begin{proof}
If $x \notin X$, then after deleting $X$ from $G$ (which means deleting at most $k$ vertices)
there would exist indices $i \neq j$ such that no vertex from $P_i$, $P_j$, $B_i$, and $B_j$ was deleted.
But then the disjoint blocks $B_i$ and $B_j$ were connected by the path $P_i-x-P_j$,
and by the previous lemma, this is a contradiction.
\qed
\end{proof}

We can decide whether a vertex $v$ is well-attached
in time $O(f'(k)e)$ using standard flow techniques, where $e=|E(G)|$.
This can be done by simply testing for each possible set of
$k+2$ disjoint blocks whether there exist the required disjoint paths that
lead from $x$ to these blocks.
Since the number of blocks in $R$ depends only on $k$, and we can find $p$ disjoint paths
starting from a given vertex of a graph $G$
in time $O(p|E(G)|)$, we can observe that this can be done indeed in time $O(f'(k)e)$.

\textbf{Finding flat zones.}
Now we show that if there are no well-attached vertices in the graph $G$,
then a flat zone exists in our grid.

\begin{lemma}
\label{comp_block}
Let $X \in \textup{ApexSets}(G,k)$, and let $G$ not include any well-attached vertices.
If $K$ is a grid-component, then there cannot exist $(k+1)^2$ disjoint blocks such that
$K$ is attached to an inner vertex of each block.
\end{lemma}

\begin{proof}
Let us assume for contradiction that there exist $(k+1)^2$ such blocks.
Since $|X|\leq k$, at least $(k+1)^2-k$ of these blocks do not contain any vertex of $X$.
So let $x_1$, $x_2$, \dots $x_{(k+1)^2-k}$ be adjacent to $K$ and let $B_1, B_2, \dots, B_{(k+1)^2-k}$
be disjoint blocks of $G-X$ such that $x_i$ is an inner vertex of $B_i$.

Since $G-X$ is planar, it follows from Lemma \ref{blocks}
that a component of $K-X$ cannot be adjacent to different vertices from $\{x_i | 1 \leq i \leq (k+1)^2-k \}$.
So let $K_i$ be the connected component of $K-X$ that is attached to $x_i$ in $G-X$.
$K$ is connected in G, hence
for every $K_i$ there is a vertex of $T=K \cap X$ that is adjacent to it in $G$.
Since there are no well-attached vertices in $G$,
every vertex of $T$ can be adjacent to at most $k+1$ of these subgraphs.
But then $|T| \geq ((k+1)^2-k)/(k+1) > k$ which is a contradiction since $T \subseteq X$.
\qed
\end{proof}

Let us now fix the constant $d=(k+1)((k+1)^2-1)$.

\begin{lemma}
\label{comp_block2}
Let $X \in \textup{ApexSets}(G,k)$, let $G$ not include any well-attached vertices,
and let $x$ be a vertex of the grid $R$.
Then there cannot exist $B_1, B_2, \dots, B_{d+1}$ disjoint blocks such that for all $i$
$( 1 \leq i \leq d+1)$ an inner vertex of $B_i$ and $x$ are both attached to
some grid-component $K_i$.
\end{lemma}

\begin{proof}
As a consequence of Lemma \ref{comp_block}, each of the grid-components $K_i$
can be attached to at most $(k+1)^2-1$ disjoint blocks.
But since $x$ is not a well-attached vertex, there can be only at most
$k+1$ different grid-components among the grid-components $K_i$, $1 \leq i \leq d+1$.
So the total number of disjoint blocks that are
attached to $x$ through a grid-component is at most $(k+1)((k+1)^2-1)=d$.
\qed
\end{proof}

\begin{lemma}
\label{zones}
Let $X \in \textup{ApexSets}(G,k)$, and let $G$ not include any well-attached vertices.
Then there exists a flat zone $Z$ in $G$.
\end{lemma}

\begin{proof}
Let $Z \in \mathcal{Z}$ be an open zone which has a vertex $w$ in its core that is attached to a vertex
$v$ of another zone in $\mathcal{Z}$ ($v \notin V(Z)$) through a grid-component $K$.
By the choice of the size of the zones and their cores, we have disjoint blocks $B_w$ and $B_v$
containing $w$ and $v$ respectively as inner points.
We can also assume that $B_w$ is a subgraph of $Z$
which does not intersect the outer circle of $Z$.

By Lemma \ref{blocks} we know that $B_w$, $B_v$ or $K$ contains a vertex from $X$.
Let $\mathcal{Z}_1$ denote the set of zones in $\mathcal{Z}$
with an inner vertex in $X$,
let $\mathcal{Z}_2$ denote the set of open zones in $\mathcal{Z}$ with a
core vertex to which a grid-component, having a common vertex with $X$, is attached,
and finally let $\mathcal{Z}_3$ be the set of the remaining open zones in $\mathcal{Z}$.
Since $|X| \leq k$ and a grid-component can be attached to inner vertices of at most
$(k+1)^2$  disjoint blocks by Lemma \ref{comp_block},
we have that $|\mathcal{Z}_1|\leq k$ and $|\mathcal{Z}_2|\leq k(k+1)^2$.

Let us count the number of zones in $\mathcal{Z}_3$.
To each zone $Z$ in $\mathcal{Z}_3$ we assign a vertex $u(Z)$ of the grid not in $Z$,
which is connected to the core of $Z$ by a grid-component.
First, let us bound the number of zones $Z$ in $\mathcal{Z}_3$ for which $u(Z) \in X$.
Lemma \ref{comp_block2} implies that any $v \in X$ can be connected this way to at most $d$ zones, so
we can have only at most $kd$ such zones.

Now let $U= \{ v \ | \ v=u(Z), Z \in \mathcal{Z}_3 \}$.
Let $a$ and $b$ be different members of $U$, and let $a$ be connected through the grid-component $K_a$
with the core vertex $z_a$ of $Z_a \in \mathcal{Z}_3$.
Let $B_a$ denote a block which only contains vertices that are inner vertices of $Z_a$,
and contains $z_a$ as inner vertex. Such a block can be given due to the size of a zone and its core.
Let us define $K_b$, $z_b$, $Z_b$, and $B_b$ similarly.
Note that $V(B_a) \cap X = V(B_b) \cap X= \emptyset$ by $Z_a,Z_b \notin \mathcal{Z}_1$.

Now let us assume that $a$ and $b$ are in the same component of $R-X$. Let $P$ be a path
connecting them in $R-X$. If $P$ has common vertices with $B_a$ (or $B_b$) then we modify $P$
the following way. If the first and last vertices reached by $P$ in $Z_a$ (or $Z_b$, resp.)
are $w$ and $w'$, then we swap the $w - w'$ section of $P$
using the outer circle of $Z_a$ (or $Z_b$, resp.).
This way we can fix a path in $R-X$ that connects $a$ and $b$, and does not include any vertex
from $B_a$ and $B_b$. But this path together with $K_a$ and $K_b$ would yield a path in $G-X$
that connects two inner vertices of $B_a$ and $B_b$, contradicting Lemma \ref{blocks}.

Therefore, each vertex of $U$ lies in a different component of $R-X$.
But we can only delete at most $k$
vertices, and each vertex in a hexagonal grid has at most 3 neighbors, thus we can conclude
that $|U| \leq 3k$. As for different zones $Z_1$ and $Z_2$ in $\mathcal{Z}$ we cannot have $u(Z_1)=u(Z_2)$
(which is also a consequence of Lemma \ref{blocks}) we have that $|\mathcal{Z}_3| \leq 3k$.
So if we choose the number of zones in $\mathcal{Z}$ to be
$z=7k+k(k+1)^2+kd+1$ we have that there are at least $3k+1$ zones in $\mathcal{Z}$
which are not contained in $\mathcal{Z}_1 \cup \mathcal{Z}_2 \cup \mathcal{Z}_3$,
indicating that they are closed.
Since a vertex can be contained by at most 3 zones, $|X| \leq k$ implies that
there exist a closed zone $Z^* \in \mathcal{Z}$, which does not contain any vertex from $X$,
and all grid-components attached to $Z^*$ are also disjoint from $X$.
This immediately implies that $T(Z^*)$ is a subgraph of $G-X$, and thus $T(Z^*)$ is planar.
\qed
\end{proof}

\textbf{Algorithm for Phase I.}
The exact steps of Phase I of the algorithm $\mathcal{A}$ are
shown in Fig. \ref{alg_phase_I}.
It starts with running algorithm $\mathcal{B}$ on the graph $G$ and integers
$w(r,k)$ and $r$.
If $\mathcal{B}$ returns a hexagonal grid as a topological minor, then the algorithm
proceeds with the next step.
If $\mathcal{B}$ returns a tree decomposition $\mathcal{T}$ of width $w(r,k)$, then Phase I returns
the triple $(G,W,\mathcal{T})$. Otherwise $G$ does not have $H_r$ as minor and its
treewidth is larger than $w(r,k)$, so by Lemma \ref{alg_b} we can conclude that $G \notin \textup{Apex}(k)$.

\begin{figure}
\begin{center}
\fbox{ \parbox{11cm}{
\textbf{Phase I of algorithm  \mbox{\boldmath $\mathcal{A}$:}} \\[4pt]
\begin{tabular}{ll}
& Input: $\ G=(V,E)$.\\[4pt]
& Let $W=\emptyset$.\\[4pt]
\end{tabular}\\
\begin{tabular}{llp{10.3cm}}
& 1. \hspace{-2pt} &
Run algorithm $\mathcal{B}$ on $G$, $w(r,k)$, and $r$. \\
& & If it returns a subgraph $H'_r$ topologically isomorphic to $H_r$ then go to Step 2.
If it returns a tree decomposition $\mathcal{T}$ of $G$, then
output($G$, $W$, $\mathcal{T}$). Otherwise output(``No solution.'').\\[4pt]
& 2. \hspace{-2pt} &
For all zones $Z$ do: \\
& & If $Z$ is flat then $G:=G-T(R_0)$, and go to Step 1. \\[4pt]& 3. \hspace{-2pt} &
Let $U=\emptyset$. For all $x \in V$: if $x$ is well-attached then $U:=U \cup \{x\}$. \\
& & If $|U| = \emptyset$ or $|W|+|U|>k$ then output(``No solution.''). \\
& & Otherwise $W:=W \cup U$, $G:=G-U$ and  go to Step 1. \\[4pt]
\end{tabular}
}}
\end{center}
\caption{Phase I of algorithm $\mathcal{A}$.}
\label{alg_phase_I}
\end{figure}

In the next step the algorithm tries to find a flat zone $Z$.
If such a zone is found, then
the algorithm executes a deletion, whose correctness is implied by
Lemma \ref{equivalence}. Note that after altering the graph,
the algorithm must find the hexagonal grid again and thus has to run $\mathcal{B}$ several times.

If no flat zone was found in Step 2, the algorithm removes well-attached
vertices from the graph in Step 3.
The vertices already removed this way are stored in $W$, and $U$ is
the set of vertices to be removed in the actual step.
By Lemma \ref{well-attached}, if $X \in \textup{ApexSets}(G,k)$ then
$W \cup U \subseteq X$, so $|W|+|U|>k$ means that there is no solution.
By Lemma  \ref{zones}, the case $U = \emptyset$ also implies $G \notin \textup{Apex}(k)$.
In these cases the algorithm stops with the output `` No solution.''
Otherwise it proceeds with updating the variables $W$ and $G$, and continues with Step 1.

The output of the algorithm can be of two types:
it either refuses the instance (outputting ``No solution.'')
or it returns an instance for Phase II.
For the above mentioned purposes the new instance is equivalent with the original problem instance
in the following sense:

\begin{theorem}
\label{kernel}
Let $(G',W,\mathcal{T})$ be the triple returned by $\mathcal{A}$ at the end of Phase I.
Then for all $X \subseteq V(G)$ it is true that $X \in \textup{ApexSets}(G,k)$ if and only if
$W \subseteq X$ and $(X \setminus W) \in \textup{ApexSets}(G',k-|W|)$.
\end{theorem}

Now let us examine the running time of this phase.
The first step can be done in time $O(f''(k)n)$ according to \cite{sey94,bod96,per00}
where $n=|V(G)|$.
Since the algorithm only runs algorithm
$\mathcal{B}$ again after reducing the number of the vertices in $G$,
we have that $\mathcal{B}$ runs at most $n$ times.
This takes $O(f''(k)n^2)$ time.
The second step requires only linear time (a breadth first search and a planarity test).
Deciding whether a vertex is well-attached can be done in time $O(f'(k)e)$ (where $e=|E(G)|$),
so we need $O(f'(k)ne)$ time to check every vertex at a given iteration in Step 3.
Note that the third step is executed at most $k+1$ times, since at each iteration $|W|$ increases.
Hence, this phase of algorithm $\mathcal{A}$ uses
total time $O(f''(k)n^2+f'(k)kne)=O(f(k)n^2)$, as the number of edges is $O(kn)$.

\section{Phase II of Algorithm $\mathcal{A}$}
\label{phase_II}

At the end of Phase I of algorithm $\mathcal{A}$ we either conclude that $G \notin \textup{Apex}(k)$,
or we have a triple $(G',W,\mathcal{T})$ for which Theorem \ref{kernel} holds. Here $\mathcal{T}$
is a tree decomposition for $G'$ of width at most $w(r,k)$.
This bound only depends on $r$ which is a function of $k$.
From the choice of the constants $r,q,z$, and $d$ we can derive by a straightforward calculation that
$\textup{tw}(G') \leq w(r,k) \leq 100 (k+2)^{7/2}$.

In order to solve our problem, we only have to find out if there is a set
$Y \in \textup{ApexSets}(G',k')$
where $k'=k-|W|$. For such a set, $Y \cup W$ would yield a solution for the original
$k$-\textsc{Apex} problem.

A theorem by Courcelle states that every graph property defined by a formula in monadic
second-order logic (MSO) can be evaluated in linear time if the input graph has bounded treewidth.
Here we consider graphs as relational structures of vocabulary $\{V,E,I\}$,
where $V$ and $E$ denote unary relations interpreted as the vertex set and the edge set of the graph,
and $I$ is a binary relation interpreted as the incidence relation. For instance,
a formula stating that $x$ and $y$ are neighboring vertices is the following: $\exists e: Ixe \land Iye$.
We will denote by $U^G$ the universe of the graph $G$, i.e.,  $U^G = V(G) \cup E(G)$.
Variables in monadic second-order logic can be element or set variables, and the
containment relation between an element variable $x$ and a set variable $X$ is simply expressed by the formula $Xx$.
For a survey on MSO logic on graphs see \cite{cou97}.

Following Grohe \cite{gro04}, we use a strengthened version of Courcelle's Theorem:

\begin{theorem} \emph{(\cite{flu02})}
\label{mso}
Let
$\varphi(x_1, \dots, x_i, X_1, \dots, X_j, y_1, \dots, y_p, Y_1, \dots, Y_q)$
denote an MSO-formula and let $w \geq 1$.
Then there is a linear-time algorithm that, given a graph  $G$ with \emph{$\textup{tw}(G) \leq w$}
and $b_1, \dots, b_p \in U^G, B_1, \dots, B_q \subseteq U^G$, decides whether
there exist $a_1, \dots, a_i \in U^G, A_1, \dots, A_j \subseteq U^G$ such that
\begin{center}
\vspace{-6pt}
$G \vDash \varphi(a_1, \dots, a_i, A_1, \dots, A_j, b_1, \dots, b_p, B_1, \dots, B_q)$,
\vspace{-6pt}
\end{center}
and, if this is the case, computes such elements $a_1, \dots , a_i$ and sets $A_1, \dots, A_j$.
\end{theorem}

It is well-known that there is an MSO-formula $\varphi_{\textrm{planar}}$ that describes the planarity of graphs,
i.e., for every graph $G$ the statement $G \vDash \varphi_{\textrm{planar}}$ holds if and only if $G$ is planar.
The following simple claim shows that we can also create a formula describing the $\textup{Apex}(k)$ graph class.

\begin{theorem}
\label{planar_mso}
There exists an MSO-formula $\textup{apex}(x_1, \dots, x_{k'})$ for which the statement
$G \vDash \textup{apex}(v_1, \dots, v_{k'})$ holds if and only if
$\{v_1, \dots, v_{k'} \} \in \textup{ApexSets}(G,k)$.
\end{theorem}

\begin{proof}
We will use the simple characterization of planar graphs by Kuratowski's Theorem:
a graph is planar if and only if it does not contain any subgraph topologically isomorphic to $K_5$ or $K_{3,3}$.
To formulate the existence of these subgraphs as an MSO-formula, we need some more simple formulas.

First, it is easy to see that the following formula expresses the property that $(X,Y)$ is a partition of the set $Z$:
$$\textup{partition}(X,Y,Z) := \forall z:
\left( Zz  \rightarrow
\left(
\left( Xz \rightarrow \neg Yz \right) \land \left( \neg Xz \rightarrow Yz \right)
\right)
\right)
$$
Using this, we can express that the vertex set $Z$ contains a path connecting $a$ and $b$, by saying that
every partition of $Z$ that separates $a$ and $b$ has to separate two neighboring vertices:
\begin{eqnarray*}
& &
\textup{connected}(a,b,Z) := Za \land Zb \land \forall X \forall Y: \\
& & \quad \left(
\left( \textup{partition}(X,Y,Z) \land Xa \land Yb  \right) \rightarrow
\left( \exists c \exists d \exists e: Xc \land Yd \land Ice \land Ide \right)
\right)
\end{eqnarray*}
The following two formulas express that two sets are disjoint, or their intersection is some given unique vertex.
\begin{eqnarray*}
& &
\textup{disjoint}(X,Y) := \forall z: (Xz \rightarrow \neg Yz) \\
& &
\textup{almost-disjoint}(X,Y,a) := \forall z: (Xz \rightarrow (\neg Yz \lor (z=a)))
\end{eqnarray*}

Now, we can state formulas expressing that a given subgraph has $K_5$ or $K_{3,3}$ as a topological minor.
For brevity, we only give the formula which states that there is a subdivision of $K_5$ in the graph
such that the vertices $v_1, v_2, \dots, v_5$ correspond to the vertices of the $K_5$, and the
vertex sets $P_{12}, P_{13}, \dots, P_{45}$ contain the subdivisions of the corresponding edges of $K_5$.
\begin{eqnarray*}
& &
\textup{$K_5$-top-minor }(v_1,v_2, \dots, v_5, P_{12},P_{13},\dots, P_{45}) :=  \\
& &
\quad \textup{connected}(v_1,v_2,P_{12}) \land \dots \land \textup{connected}(v_4,v_5,P_{45}) \land \\
& &
\quad \textup{almost-disjoint}(P_{12},P_{13},v_1) \land \dots \land \textup{almost-disjoint}(P_{35},P_{45},v_5) \land \\
& &
\quad \textup{disjoint}(P_{12},P_{34}) \land \dots \land \textup{disjoint}(P_{23},P_{45})
\end{eqnarray*}
The formula $\textup{$K_{3,3}$-top-minor}$ can be similarly created. Now, we are ready to give the $\textup{apex}$ formula,
which uses the fact that $G-X$ is planar if and only if every subdivision of $K_5$ or $K_{3,3}$ in $G$ must involve at least one vertex from $X$.
\begin{eqnarray*}
& &
\textup{apex}(v_1,v_2, \dots, v_{k'}) :=  \\
& &
\quad \forall x_1 \forall x_2 \dots \forall x_5 \forall X_1 \forall X_2 \dots \forall X_{10}:
(\textup{$K_5$-top-minor}(x_1, \dots, x_5,X_1, \dots, X_{10}) \\
& &
\quad \quad \rightarrow ((x_1=v_1) \lor \dots \lor (x_5=v_{k'})  \lor X_1 v_1 \lor \dots \lor X_{10} v_{k'})) \land \\
& &
\quad \forall x_1 \forall x_2 \dots \forall x_6 \forall X_1 \forall X_2 \dots \forall X_9:
(\textup{$K_{3,3}$-top-minor}(x_1,\dots, x_6,X_1, \dots, X_9) \\
& &
\quad \quad \rightarrow((x_1=v_1) \lor \dots \lor (x_6=v_{k'})  \lor X_1 v_1 \lor \dots \lor X_{9} v_{k'})) \\
\end{eqnarray*}
\qed
\end{proof}

Now let us apply Theorem \ref{mso}.
Let $\mathcal{C}$ be the algorithm which, given a graph $G$ of bounded treewidth,
decides whether there exist $v_1, \dots, v_{k'} \in U^G$ such that
$G \vDash \textup{apex}(v_1, \dots, v_{k'})$ is true, and if possible, also produces such variables.
By Theorem \ref{planar_mso}, running $\mathcal{C}$ on $G'$ either returns a set of vertices
$U \in \textup{ApexSets}(G', k')$, or reports that this is not possible.
Hence, we can finish algorithm $\mathcal{A}$ in the following way:
if $\mathcal{C}$ returns $U$ then output($U \cup W$), otherwise
output(``No solution'').

The running time of Phase II is $O(g(k)n)$ for some function $g$.

\begin{remark}
Phase II of the algorithm can also be done by applying dynamic programming,
using the tree decomposition $\mathcal{T}$ returned by $\mathcal{B}$.
This also yields a linear-time algorithm, with a double exponential dependence on $\textup{tw}(G')$ (and hence on $k$).
Since the proof is quite technical and detailed, we omit it.
\end{remark}

\end{document}